\documentclass[sigconf,nonacm,acmthm=false]{acmart}
\usepackage{amsmath, amsthm}
\usepackage{graphicx}
\usepackage{bbm}
\usepackage{booktabs}

\newcommand{\hookItems}{100}
\newcommand{\hookPLight}{0.01}
\newcommand{\hookPHeavy}{0.20}
\newcommand{\hookLiftAnalytic}{1.82}
\newcommand{\fmStatBlock}{14}
\newcommand{\mltwfDevShare}{0.73}
\newcommand{\mltwfProfAct}{0.067}
\newcommand{\mltwfGroups}{40}
\newcommand{\mltwfLastGroupMaxUsers}{102}
\newcommand{\lfmDevShare}{0.00}
\newcommand{\lfmVarRatio}{0.73}
\newcommand{\lfmNonposPearson}{257}
\newcommand{\lfmBlockPairs}{595}
\newcommand{\amzDevShare}{0.50}
\newcommand{\amzProfAct}{0.050}
\newcommand{\amzNonposPearson}{8}
\newcommand{\amzBlockPairs}{703}
\newcommand{\fmItems}{1{,}763}
\newcommand{\fmMinUsers}{50}
\newcommand{\fmPairs}{1.55\,M}
\newcommand{\fmDenseShare}{95\%}
\newcommand{\fmNeverPct}{1.2\%}
\newcommand{\fmIndWithinTwo}{35\%}
\newcommand{\fmActWithinTwo}{81\%}
\newcommand{\fmActSpearman}{0.31}
\newcommand{\fmRareBucketUsers}{1--2}
\newcommand{\fmRareBucketItems}{252}
\newcommand{\fmRareBucketShare}{0.1\%}
\newcommand{\fmRareWithinAct}{9.2}
\newcommand{\fmTailItems}{1{,}770}
\newcommand{\fmRareEffUsers}{3.3}
\newcommand{\fmRareContribUsers}{43}
\newcommand{\catalogGroups}{40}
\newcommand{\trendRtwoMin}{0.88}
\newcommand{\trendRtwoMax}{1.00}

\newtheorem{assumption}{Assumption}
\newtheorem{proposition}{Proposition}
\newtheorem{corollary}{Corollary}
\theoremstyle{definition}

\newtheorem{example}{Example}
\theoremstyle{remark}

\newcommand{\PP}{\mathbb{P}}
\newcommand{\EE}{\mathbb{E}}
\newcommand{\RR}{\mathbb{R}}
\newcommand{\Var}{\operatorname{Var}}
\newcommand{\Cov}{\operatorname{Cov}}
\newcommand{\lift}{\operatorname{lift}}
\newcommand{\logit}{\operatorname{logit}}
\newcommand{\logistic}{\operatorname{logistic}}
\newcommand{\x}{\mathbf{x}}
\newcommand{\X}{\mathbf{X}}
\newcommand{\indep}{\rotatebox[origin=c]{90}{$\models$}}
\newcommand{\chara}{\mathbbm 1}

\title{Do Similarity Measures Measure Similarity? How User Activity Creates Phantom Correlation in Recommender Systems}
\author{Marco Favier}
\affiliation{%
  \institution{University of Antwerp}
  \country{Belgium}}
\author{Noah Dani\"els}
\affiliation{%
  \institution{University of Antwerp}
  \country{Belgium}}
\author{Jens Leysen}
\affiliation{%
  \institution{University of Antwerp}
  \country{Belgium}}
\author{Sander Marinus}
\affiliation{%
  \institution{University of Antwerp}
  \country{Belgium}}
\author{Bart Goethals}
\affiliation{%
  \institution{University of Antwerp}
  \country{Belgium}}
\begin{document}

\begin{abstract}
Item-item similarity is one of the oldest and most foundational signals in recommender systems: two items are assumed to be similar when the same users consume them. Despite the sparsity and scale of interaction data, the implicit assumption across the field is that item co-occurrence reflects genuine shared preference—capturing what users like together.
In this paper, we demonstrate that a substantial fraction of observed item-item similarity is an artifact of variance in user activity rather than shared taste. Even if every user consumed items completely independently—differing only in their total consumption volume—standard similarity measures still detect strong associations between items. We term this phenomenon \emph{phantom correlation}: spurious item association that emerges strictly from heterogeneous user activity in the complete absence of preference structure.
Starting from minimal formal assumptions, we show that independent consumption under varying activity is described by a Rasch model, and we derive how the resulting association depends on item popularity for each similarity measure. By comparing our theoretical baseline against real-world benchmark datasets, we reveal the extent to which popular similarity measures reflect phantom correlation, which bears on how item co-occurrence is interpreted and used in collaborative filtering.
\end{abstract}


\maketitle

\section{Introduction}\label{sec:intro}

Item-item similarity is one of the oldest signals in recommender systems: two items are considered similar when the same users consume them~\cite{sarwar2001,linden2003}. It scales to large catalogs, and it stays reliable on sparse data, where interactions among more than two items are rarely observed~\cite{sarwar2001}. For these reasons, item-item signals remain in use in many of the most successful recommender systems, and they are often used as a baseline to compare more complex models~\cite{linden2003,ferraridacrema2019}. Item-item similarity drives recommendations on Web platforms such as streaming services, online shops and social media, where it is computed from logs of implicit feedback whose users differ widely in activity: understanding what these similarities measure is therefore a question about the Web data that recommender systems are built on. 

In practice, a user's history is multiplied by an item-item similarity matrix to score and rank the items. The simplest example is item-kNN~\cite{sarwar2001}, which scores an item by adding up its similarities with the items in the user's history; SLIM~\cite{ning2011} learns a sparse similarity matrix by regression, and EASE~\cite{steck2019} learns a dense one in closed form. The idea behind item-item similarity is pretty intuitive: if two items are consumed together more often than expected by chance, it is because users like them together, and a similarity measure quantifies this information~\cite{agrawal1993,brin1997,sarwar2001}.

Another important aspect of the typical data of recommender systems is that the distribution of user activity is extremely heterogeneous: some users consume a handful of items, while others consume hundreds or thousands. The distribution of the history length of users is often described as heavy-tailed, and often modeled as a power law~\cite{clauset2009}.
We show that a large part of the observed similarity can arise from these differences in activity alone, not from users liking items together. We call this phenomenon \emph{phantom correlation}: like Simpson's paradox~\cite{simpson1951} and confounding in recommendation~\cite{xu2025}, it is association created by a hidden common factor, here activity. To better illustrate this, consider the following example:

\begin{example}\label{ex:hook}
A small streaming service comprises \hookItems{} films. Users on the platform like the movies equally, so no movie is preferred over another. But the activity differs between users: half of them are \emph{casual viewers}, who watch any given film with probability $\hookPLight$. The other half are \emph{binge watchers}, who watch any given film with probability $\hookPHeavy$. Every user picks films completely independently: watching one film says nothing about which other film they will watch.

Trying to recommend films to users is futile: every film is equally good, and no two films are related to each other. Yet the recommender sees a different picture. If one computes the similarities between items, every pair turns out to be positively correlated. More precisely, movies appear together about $\hookLiftAnalytic$ times more often than chance would predict.
\end{example}

This might seem counter-intuitive at first. There is basically no information about taste in the data, so how can it be that every pair of films is positively correlated? The main reason is that the similarity measure (which in this case is lift) is not predicting the taste of the user, but rather how much they watch. Knowing that you have seen any film makes you more likely to be a binge watcher, and so more likely to see any other film.

The phantom correlation phenomenon arises because three conditions coexist: 
\begin{enumerate}
    \item Each user consumes items independently: nobody has a taste for pairs.
    \item A more active user is more likely to consume \emph{every} item.
    \item The only difference between users is how active they are.
\end{enumerate}
With identical items this is easy to model; with items of different popularity, we first have to formalize what it means to be ``the same as a casual viewer, only more active''.

We formalize the three conditions, show that they lead to a unique model, derive its consequences for similarity measures, and measure how much of the association in real datasets it explains. The model we obtain is known in psychometrics as the Rasch model~\cite{rasch1960,andersen1977,cressie1983}, and heterogeneous activity is known to bias covariance and correlation in bipartite networks~\cite{puccio2019}. Our contributions are a derivation of the model from three assumptions on consumption, the exact dependence of the induced association on item popularity for lift, cosine and Pearson correlation, and a measurement of how much of the observed association activity alone explains, also for a learned model (EASE).

\section{Related Work}\label{sec:related}
\paragraph{Item-item similarity.} Item-based collaborative filtering scores items by their similarity to the items a user consumed, measured by cosine, conditional probability or lift~\cite{sarwar2001,deshpande2004}, or learned by regression, as a sparse matrix in SLIM~\cite{ning2011} or a dense one in EASE~\cite{steck2019}. Heterogeneity has long been handled by adjusting the measures: inverse user frequency down-weights items that many users liked~\cite{breese1998}, and PMI is known to favor rare items~\cite{levy2015}. We instead derive how much association differences in activity create, and how it depends on popularity for each measure.

\paragraph{Popularity bias.} Recommenders tend to over-recommend popular items~\cite{jannach2015,klimashevskaia2024}, and this bias is usually studied in what a recommender outputs. We show that part of the popularity structure of item-item similarity is already in the data, before any learning, as a consequence of differences in activity.

\paragraph{Latent variable models.} The activity model is the Rasch model of psychometrics~\cite{rasch1960,andersen1977}, here with a nonparametric distribution of the latent trait~\cite{lindsay1991}; Cressie and Holland~\cite{cressie1983} characterize the data it can produce, and positive association holds in every unidimensional model with monotone item curves~\cite{holland1986}. Latent factor models of collaborative filtering~\cite{hofmann2004,koren2009} use several dimensions to encode taste; we use a single one on purpose, to isolate what activity alone does.

\paragraph{Bipartite networks.} Saracco et al.~\cite{saracco2017} project user-item networks onto items by keeping the pairs whose number of common users is unlikely under the bipartite configuration model, the maximum-entropy model with fixed expected degrees~\cite{park2004}. Its distribution is that of the activity model with one fixed parameter per user, so their expected co-occurrence already contains the association we study. We describe activity instead by a distribution over the population, estimated from all history lengths, and make the association it creates the object of study rather than a baseline to subtract. Degree-normalized projections~\cite{zhou2007} and swap randomization~\cite{gionis2007} also account for degrees, without describing the association that degrees induce. Closest to our work, Puccio et al.~\cite{puccio2019} show that heterogeneous degrees bias the sample covariance and correlation in bipartite networks, positively and depending on the degrees of the two nodes, and propose corrected estimators; their analysis uses two classes of nodes and small heterogeneity, and does not cover lift, cosine or learned models.

\section{Setting and Notation}\label{sec:setting}
We consider $U$ users and $n$ items, in the setting of implicit feedback: we observe which items each user interacted with (watched, listened to, reviewed or rated highly), not how much they liked them. We write $[n]:=\{1,\dots,n\}$ for the set of items and $[U]:=\{1,\dots,U\}$ for the set of users. User $u\in[U]$ is described by a binary vector $\X_u=(X_{u1},\dots,X_{un})\in\{0,1\}^n$, where $X_{ui}=1$ if user $u$ consumed item $i$ and $X_{ui}=0$ otherwise; a $0$ mixes items the user did not want with items they never came across~\cite{hu2008,liang2016}. The popularity of item $i$ is $p_i=\PP(X_i=1)$, and $p_{ij}=\PP(X_i=1,X_j=1)$ is the probability that a user consumed both $i$ and $j$. When we drop the user index, the vector $\X$ denotes a generic user.

Item-item similarity measures are functions of $p_i$, $p_j$ and $p_{ij}$. We consider the following four widely used measures~\cite{sarwar2001,deshpande2004}:
\begin{itemize}
\item \textbf{Covariance}: $\Cov(X_i,X_j):=p_{ij}-p_ip_j$.
\item \textbf{Lift}: $\lift_{ij}:=\dfrac{p_{ij}}{p_ip_j}$; its logarithm is the pointwise mutual information~\cite{church1989}.
\item \textbf{Cosine similarity}: $\cos_{ij}:=\dfrac{p_{ij}}{\sqrt{p_ip_j}}$.
\item \textbf{Pearson correlation}: $\rho_{ij}:=\dfrac{p_{ij}-p_ip_j}{\sqrt{p_i(1-p_i)\,p_j(1-p_j)}}$.
\end{itemize}
Lift equals $1$, and covariance and the Pearson correlation equal $0$, exactly when the two items co-occur as often as chance would predict. Notice that, once the popularities are fixed, all four are functions of $p_{ij}$ alone: a model that reproduces the popularities and the co-occurrences reproduces every measure.
Moreover, we will consider the similarity matrix generated by EASE~\cite{steck2019}, which is a dense item-item similarity matrix learned in closed form by linear regression. The EASE similarity matrix is not a function of $p_i$, $p_j$ and $p_{ij}$ alone, but depends on the co-occurrences and popularities of all items.

\section{The Activity-Driven Model}\label{sec:model}
In order to understand how much user activity alone can create phantom correlation and influence similarity measures, we need a model where the only mutual information between items is derived from user activity. To do so, we formalize the three conditions identified in the introduction:
\begin{enumerate}
  \item Each user consumes items independently.
  \item A more active user is more likely to consume every item.
  \item The only difference between users is how active they are.
\end{enumerate}
We assume that each user has an unobserved trait $\theta_u\in\RR$, which represents how active the user is: the larger $\theta_u$, the more items the user is expected to consume.
This unobserved trait is sufficient to formally rephrase the three conditions above:
\begin{assumption}[Independent Consumption]\label{as:independence}
Given $\theta_u$, the user consumes items independently:
\begin{equation}\label{eq:bernoulli}
\PP(\X_u=\x\mid\theta_u)=\prod_{i\in[n]} \PP(X_{ui}=x_i\mid\theta_u),
\end{equation}
with $0<\PP(X_{ui}=1\mid\theta_u)<1$ for every item $i$.
\end{assumption}

In words, the user has no taste for combinations of items or, more precisely, the trait $\theta_u$ fully characterizes the user, since given $\theta_u$ the items behave independently.

On its own, this assumption is weak: a rich enough $\theta_u$ can describe almost any behavior~\cite{hofmann2004,koren2009}. The other assumptions limit what $\theta_u$ can represent.

\begin{assumption}[Activity Influences Every Item]\label{as:monotonicity}
For every item $i\in[n]$, the probability $\PP(X_{ui}=1\mid\theta_u)$ is a measurable function of $\theta_u$ that is non-decreasing almost surely. Moreover, if $p_i\leq p_j$, then $\PP(X_{ui}=1\mid\theta_u)\leq \PP(X_{uj}=1\mid\theta_u)$ almost surely.
\end{assumption}

It makes $\theta_u$ a measure of activity, with no preference for specific items, and it preserves the order of popularity at every activity level. Many models satisfy the first two assumptions, and all of them make every pair of items non-negatively correlated (Proposition~\ref{prop:positive}). The third assumption makes the model unique.

\begin{assumption}[Activity Is the Only Difference]\label{as:volume}
Once we know how many items a user consumed, $S_u:=\sum_{i\in[n]}X_{ui}$, the parameter $\theta_u$ does not carry further information. Formally,
\[
  \X_u \indep \theta_u \mid S_u.
\]
Equivalently, the count $S_u$ is a sufficient statistic for the parameter $\theta_u$.
\end{assumption}
Once we know how many items a user consumed, their activity tells us nothing more about which ones: $\PP(\X_u=\x\mid S_u)$ is the same for every user.

\subsection{Deriving the Model}
The three assumptions completely determine the shape of the model; the elementary proof shows why ``more active'' must mean ``multiplied odds''.
\begin{proposition}[The activity model]\label{prop:model}
A distribution $\PP(\X=\x)$ satisfies Assumptions~\ref{as:independence}--\ref{as:volume} if and only if there exist item parameters $h_1,\dots,h_n\in\mathbb{R}$ and a distribution $P$ of activity levels $\theta\in\mathbb{R}$ such that
\begin{equation}\label{eq:maxent}
\begin{aligned}
\PP(\X=\x)&=\int\frac{1}{Z(\theta)}\exp\Big(\sum_{i\in[n]}(h_i+\theta)\,x_i\Big)\,dP(\theta),\\
Z(\theta)&=\prod_{i\in[n]}\big(1+e^{h_i+\theta}\big).
\end{aligned}
\end{equation}
In particular, given the activity level $\theta$, it holds that $\logit\PP(X_i=1\mid\theta)=h_i+\theta$ for every item $i$.
\end{proposition}
\begin{proof}
($\Rightarrow$) Since $x_i\in\{0,1\}$, we have $\log\PP(X_i=x_i\mid\theta)=x_i\logit\PP(X_i=1\mid\theta)+\log\PP(X_i=0\mid\theta)$, so Assumption~\ref{as:independence} gives
\begin{equation}\label{eq:expform}
\begin{aligned}
&\PP(\X=\x\mid\theta)
=\exp\Big(\sum_{i\in[n]}\log\PP(X_i=x_i\mid\theta)\Big)\\
&=\exp\Big(\sum_{i\in[n]}\big(x_i\logit\PP(X_i=1\mid\theta)+\log\PP(X_i=0\mid\theta)\big)\Big).
\end{aligned}
\end{equation}
By the Fisher--Neyman factorization theorem~\cite{halmos1949}, Assumption~\ref{as:volume} holds if and only if $\PP(\X=\x\mid\theta)=g(S(\x),\theta)\,\varphi(\x)$ for some non-negative functions $g$ and $\varphi$, with $\varphi$ independent of $\theta$, which means
\begin{equation*}
\begin{aligned}
&\sum_{i\in[n]}\big(x_i\logit\PP(X_i=1\mid\theta)+\log\PP(X_i=0\mid\theta)\big)\\
&\qquad=\log g(S(\x),\theta)+\log\varphi(\x).
\end{aligned}
\end{equation*}
Take two baskets $\x$ and $\mathbf{y}$ that coincide except for items $j\neq k$, with $x_j=y_k=1$ and $x_k=y_j=0$ (for instance with all other items set to $0$), and the same $\theta$. They have the same size, so subtracting the two equations gives
\begin{equation*}
\logit\PP(X_j=1\mid\theta)-\logit\PP(X_k=1\mid\theta)=\log\varphi(\x)-\log\varphi(\mathbf{y}).
\end{equation*}
The right-hand side does not depend on $\theta$, while the left-hand side does not depend on any item other than $j$ and $k$. So if we define the function $\alpha(\theta):=\logit\PP(X_1=1\mid\theta)$, then for every item $i$
\begin{equation*}
\logit\PP(X_i=1\mid\theta)=h_i+\alpha(\theta)
\end{equation*}
for some constant $h_i$. The conditional distribution in Equation~\eqref{eq:expform} therefore depends on $\theta$ only through $\alpha(\theta)$, so integrating over $\theta$ or over $\alpha(\theta)$ gives the same distribution of baskets. Without loss of generality we can therefore take $\alpha$ to be the identity, replacing $\theta$ by $\alpha(\theta)$ and $P$ by the distribution of $\alpha(\theta)$ across users. Substituting $\logit\PP(X_i=1\mid\theta)=h_i+\theta$ and $\PP(X_i=0\mid\theta)=\logistic(-(h_i+\theta))$ into Equation~\eqref{eq:expform} gives
\begin{equation*}
\PP(\X=\x\mid\theta)=\frac{1}{Z(\theta)}\exp\Big(\sum_{i\in[n]}(h_i+\theta)\,x_i\Big),
\end{equation*}
and integrating against $P$ gives Equation~\eqref{eq:maxent}.

($\Leftarrow$) If Equation~\eqref{eq:maxent} holds, take the activity level $\theta$, with distribution $P$, as the trait. Given $\theta$ the items are independent, with $\PP(X_i=1\mid\theta)=\logistic(h_i+\theta)$; the trait enters only through $\theta\,S(\x)$, so the count is sufficient; and every $\logistic(h_i+\theta)$ is non-decreasing in $\theta$, with the popularity order given by the order of the $h_i$ at every activity level.
\end{proof}

To understand intuitively what the model says, we need to see how the odds of consuming an item change with the activity level. For two activity levels $\theta$ and $\theta'$,
\[
  \frac{\PP(X_i=1\mid\theta')}{\PP(X_i=0\mid\theta')}
  = e^{\theta'-\theta}\,\frac{\PP(X_i=1\mid\theta)}{\PP(X_i=0\mid\theta)}\quad\text{for every item }i:
\]
a more active user has the odds of \emph{every} item multiplied by the same factor. For a niche item, odds and probabilities are almost the same thing, so its probability is multiplied by about the same factor; for a blockbuster that is already likely, the probability barely changes. This property, together with the fact that items are independent given the activity level, fully determines the model.

The activity model in Equation~\eqref{eq:maxent} is known in psychometrics as the Rasch model~\cite{rasch1960}, where $S_u$ is the score of a test taker and is known to be a sufficient statistic; Andersen~\cite{andersen1977} proved a converse: if a minimal sufficient statistic for the person parameter exists and does not depend on the item parameters, it must be the raw score, and the model must be of Rasch type. Cressie and Holland~\cite{cressie1983} characterize the distributions of answer patterns that a Rasch model can produce, in a form equivalent to Equation~\eqref{eq:conditional} below. Our derivation is similar in spirit, but its setting differs: Andersen assumes that all items follow a single parametric family, while we let every item have its own consumption curve and assume directly that the number of consumed items is sufficient. In network science, the same conditional distribution with one parameter per user is the bipartite configuration model~\cite{saracco2017} (Section~\ref{sec:related}).

\section{How Activity Creates Correlation}\label{sec:correlation}
\subsection{Model-Independent Consequences}
\subsubsection{Total covariance depends on user activity}
While the model we have previously derived will be useful to quantify the amount of phantom correlation in real datasets, the fact that differences in activity create correlation is a more general phenomenon that holds for any distribution on $\X$. It is a consequence of the following identity:
\begin{equation}\label{eq:identity}
\sum_{i\neq j}\Cov(X_i,X_j)=\Var(S)-\sum_{i\in[n]} p_i(1-p_i),
\end{equation}
which follows from the elementary identity $\Var(X+Y)=\Var(X)+\Var(Y)+2\Cov(X,Y)$, applied to $S=\sum_{i\in[n]}X_i$. In particular, if we fix the item popularities $p_i$, the total covariance is determined by the variance of the number of items consumed by users. In other words, the more users differ in how many items they consume, the larger the total covariance between items, whatever the model.

The identity fixes only the total, not the covariance of individual pairs or other measures.
\subsubsection{Every pair is positively associated}
One might model activity differently, for instance letting a niche item reveal more about activity than a popular one~\cite{breese1998}. Assumption~\ref{as:volume} is not needed to show that activity creates correlation:
\begin{proposition}\label{prop:positive}
Under Assumptions~\ref{as:independence} and~\ref{as:monotonicity}, for every pair of items $i,j$,
\begin{equation}\label{eq:totalcov}
\Cov(X_i,X_j)= \Cov_\theta\big(\PP(X_i=1\mid\theta),\,\PP(X_j=1\mid\theta)\big)\geq 0.
\end{equation}
Moreover, the inequality is strict unless $\theta \indep X_i$ or $\theta \indep X_j$.
\end{proposition}
\noindent This is a special case of a known result for latent variable models with monotone item curves~\cite{holland1986}; the proof is in Appendix~\ref{app:proofs}.

Under the activity model, $\PP(X_i=1\mid\theta)=\logistic(h_i+\theta)$ is strictly increasing, so every pair is strictly positively correlated unless activity is constant: \emph{phantom correlation can only push correlation up}, for every single pair.

\subsubsection{Which activity patterns are possible} The two previous results can be combined to give a necessary condition for a distribution $\PP(\X)$ to be compatible with Assumptions~\ref{as:independence} and~\ref{as:monotonicity}, as expressed by the following inequality:

\begin{corollary}\label{cor:possible}
Under Assumptions~\ref{as:independence} and~\ref{as:monotonicity}, the following inequality holds:
\begin{equation}\label{eq:overdispersion}
\Var(S)\;\geq\;\sum_{i\in[n]} p_i(1-p_i).
\end{equation}
\end{corollary}

No model satisfying Assumptions~\ref{as:independence} and~\ref{as:monotonicity} can generate a dataset that violates it. The inequality is violated when every user consumes exactly the same number of items, and it can fail when the length of every list is capped, as in datasets that keep only each user's top items (Section~\ref{sec:catalog}). The following example illustrates this.
\begin{example}[Fixed baskets]\label{ex:fixed}
Suppose every user consumes exactly $5$ of $100$ identical items, chosen at random. Then $\Var(S)=0$, and Equation~\eqref{eq:overdispersion} does not hold. In fact, each pair has negative covariance $\frac{5\cdot4}{100\cdot99}-0.05^2<0$, which is impossible according to Proposition~\ref{prop:positive}. No activity distribution can produce this, since Assumption~\ref{as:independence} is not satisfied. A fixed basket size forces items to compete for the same slots.
\end{example}

Real datasets tend to sit at the opposite extreme. In recommendation data, a few heavy users typically sit next to many light ones~\cite{clauset2009}, so the history lengths vary far more than the right-hand side of Equation~\eqref{eq:overdispersion}. This is the case for the datasets we use, with a single exception (Section~\ref{sec:experiments}).

\subsection{Activity as a Parameter and Activity as a Count}\label{sec:modelresults}
Assumption~\ref{as:volume} makes the count $S=\sum_i X_i$ sufficient for the activity $\theta$: once we know how many items a user consumed, the activity tells us nothing more about the basket. It might seem intuitive to conclude that the count can replace the parameter, and that users who consumed the same number of items consume them independently. However, this is \emph{not} the case. We can in fact easily compute the conditional distribution of a basket given its size, like so:
\begin{equation}\label{eq:conditional}
\PP(\X=\x\mid S=k)=\frac{\exp\big(\sum_{i\in[n]}h_ix_i\big)}{\displaystyle\sum_{\substack{I\subseteq[n]\\|I|=k}}\exp\big(\sum_{i\in I}h_i\big)}.
\end{equation}
Importantly, this conditional distribution is identical to the conditional distribution of independent Bernoulli variables with parameters $\logistic(h_i)$, conditioned on the sum being $k$~\cite{chen1994}. The difference between the two distributions lies in the distribution $\PP(S=k)$ of the count, which in the Bernoulli case follows a Poisson-binomial distribution~\cite{hong2013}, while in the activity model it depends on the distribution of $\theta$; for the Rasch model, this representation is due to Cressie and Holland~\cite{cressie1983}. This, together with classical results on independent Bernoulli variables conditioned on their sum~\cite{joagdev1983,efron1965}, gives the following statement:
\begin{proposition}[Competition and activity]\label{prop:competition}
Under the activity model, for every pair of items $i\neq j$:
\begin{enumerate}
\item $\Cov(X_i,X_j\mid S=k)<0$ for $1\leq k\leq n-1$;
\item the probability $\PP(X_i=1\mid S=k)$ is non-decreasing in $k$.
\end{enumerate}
\end{proposition}

\subsection{Similarity Measures under the Activity Model}\label{sec:popularity}

How much of this association each measure shows, and for which pairs, depends on the popularity of the two items and differs between measures; below we derive what the matrices of Figure~\ref{fig:ml1m} should look like if activity were the only source of association.

\subsubsection{Lift}
Under the activity model, lift is the measure that gives the most direct insight into the effect of activity on association, as the following proposition shows.

\begin{proposition}\label{prop:limits}
Under the activity model, by Equation~\eqref{eq:totalcov}, the lift of two items $i$ and $j$ is
\begin{equation}\label{eq:liftmodel}
\lift_{ij}=\frac{\EE_\theta\big[\logistic(h_i+\theta)\,\logistic(h_j+\theta)\big]}{\EE_\theta\big[\logistic(h_i+\theta)\big]\;\EE_\theta\big[\logistic(h_j+\theta)\big]}.
\end{equation}
In particular, for a fixed distribution of $\theta$ with finite $\EE[e^\theta]$, the following hold:
\begin{enumerate}
\item if $p_i,p_j\to 0$, then $\lift_{ij}\to 1+\Var(e^\theta)/\EE[e^\theta]^2$;
\item if $p_i\to 1$, then $\lift_{ij}\to1$;
\item $\lift_{ij}$ is strictly decreasing in $p_i$ and in $p_j$, unless $X_i$ and $X_j$ are independent.
\end{enumerate}
\noindent The proof is in Appendix~\ref{app:proofs}.
\end{proposition}

\subsubsection{Cosine and Pearson Correlation}
Both measures are lift combined with a function of the popularities:
\begin{equation}\label{eq:fromlift}
\cos_{ij}=\lift_{ij}\sqrt{p_ip_j},\qquad
\rho_{ij}=(\lift_{ij}-1)\sqrt{\frac{p_ip_j}{(1-p_i)(1-p_j)}}.
\end{equation}
The behavior in the two limits is therefore the same as for lift, multiplied by a factor that depends on the popularities. However, monotonicity is not guaranteed, since the new factor increases with popularity while lift decreases. We will explore the behavior of these measures empirically. For covariance and Pearson correlation, Puccio et al.~\cite{puccio2019} obtained the dependence on popularity to second order in the heterogeneity, for two classes of users; Proposition~\ref{prop:limits} is exact and holds for any distribution of activity.

\subsubsection{EASE Similarity Matrix}
EASE~\cite{steck2019} learns a dense matrix of item-item weights $B$ by regressing each item on all the others, with an $\ell_2$ penalty of strength $\lambda$ and a zero diagonal. Its solution depends on the data only through the Gram matrix $G$ of the interaction matrix:
\[
P=(G+\lambda I)^{-1},\qquad B_{ij}=-\frac{P_{ij}}{P_{jj}}\ \ (i\neq j),\qquad B_{jj}=0.
\]
A model has no interaction matrix, but it has an expected Gram matrix: we replace the observed counts by the counts the model expects, computed under the fitted model (Section~\ref{sec:fitting}), and compute EASE on it. Unlike lift, cosine and Pearson correlation, a weight $B_{ij}$ depends on the co-occurrences of all items, not only on those of $i$ and $j$, so we study it empirically.
\begin{figure*}[t]
\centering
\includegraphics[width=\linewidth]{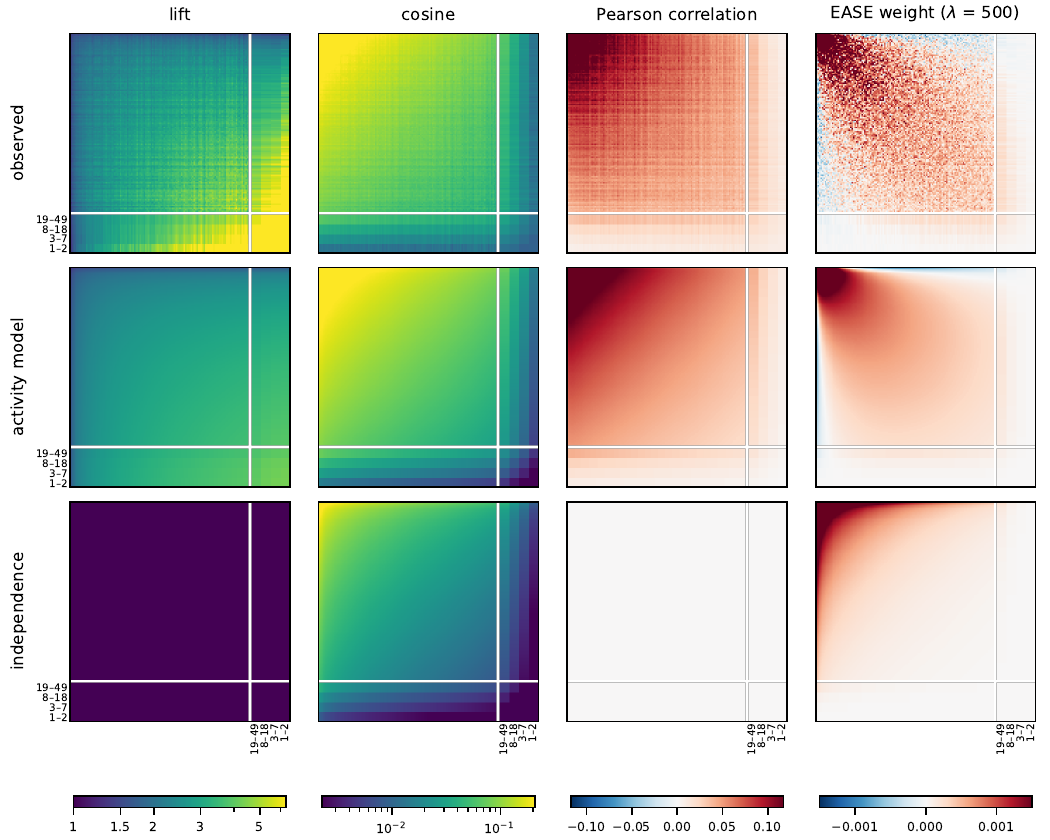}
\caption{Similarity matrices for MovieLens-1M. Rows: observed data, activity model, independence model; columns: lift, cosine, Pearson correlation, EASE weights. Films are sorted by popularity (most popular at the top left); films with fewer than 50 positive ratings are pooled into four buckets by number of users (strips at the right and bottom).}
\label{fig:ml1m}
\end{figure*}

\section{Measuring Activity-Induced Association in Data}\label{sec:fitting}
To quantify phantom correlation in real data, we fit the model by expectation-maximization~\cite{dempster1977}. By Equation~\eqref{eq:conditional}, activity enters the distribution of baskets only through $\PP(S=k)$~\cite{cressie1983}, so every distribution of $\theta$ with
\[
\PP(S=k) = \int \PP(S=k\mid \theta)\,d\PP(\theta)
\]
for every observed history length $k$ gives the same model.
We opt to approximate $\PP(\theta)$ by a discrete distribution over $G$ fixed activity levels $\theta_1<\dots<\theta_G$, that is,
\[
  \PP(\theta) := \sum_{g=1}^G \omega_g \chara(\theta=\theta_g),
\]
where $\chara$ is the indicator function and $\omega_g$ is the weight of activity level $\theta_g$. Such a discrete distribution is the nonparametric maximum-likelihood estimator of a Rasch model with random activity~\cite{lindsay1991}. Details of the initialization are given in Appendix~\ref{sec:fitdetails}.
\subsubsection{Fitting.} The activity model has two unknown parts: the vector $\mathbf{h}:=(h_1,\dots, h_n)$ of item parameters, and the weights $\boldsymbol{\omega}:=(\omega_1,\dots, \omega_G)$ of the distribution $\PP(\theta)$ of activity. 
The EM algorithm performs the following two steps iteratively until convergence.\\
\textbf{E-step:} Calculate $\mathbb{P}(\theta=\theta_g\mid S=k)$ based on the current estimates of $\mathbf{h}$ and $\boldsymbol{\omega}$. By Bayes' theorem and Equation~\eqref{eq:maxent}, summing over the baskets of size $k$,
\[
\mathbb{P}(\theta=\theta_g\mid S=k)\propto
\omega_g\exp\Big(\theta_g k-\sum_i\log\big(1+e^{h_i+\theta_g}\big)\Big),
\]
where the factor that depends on which items were consumed is the same for every $g$ and cancels in the normalization. The posterior therefore depends on a user only through the number of items they consumed.\\
\textbf{M-step:} First we update the weights $\omega_g$ using the empirical distribution $P(S=k)$ of the observed history lengths:
\[
  \omega_g\leftarrow \sum_k P(S=k)\,\mathbb{P}(\theta=\theta_g\mid S=k).
\]
Then we update the item parameters $h_i$ by solving the following equation for each item $i$:
\[
  \sum_g \omega_g \logistic(h_i+\theta_g) = p_i.
\]
In practice, instead of solving the equation exactly, we perform a single Newton step per iteration. At convergence, the model therefore reproduces every item popularity exactly.
The log-likelihood of the data that the EM algorithm maximizes is
\begin{equation}\label{eq:loglik}
{\ell(\mathbf{h},\boldsymbol{\omega})}={U}\left[\sum_i p_ih_i
+\sum_k P(S=k)\log\sum_{g=1}^G\omega_g\,\frac{e^{\theta_g k}}{Z(\theta_g)}\right],
\end{equation}
with $Z(\theta)$ as in Proposition~\ref{prop:model}. Notice that the data enter only through the item popularities $p_i$ and the distribution $P(S=k)$ of history lengths, as Equation~\eqref{eq:conditional} implies: the co-occurrences of item pairs are never used, yet the fitted model implies the co-occurrence 
\begin{equation*}
p_{ij}^{\text{act}}=\sum_g\omega_g\,\logistic(h_i+\theta_g)\,\logistic(h_j+\theta_g).
\end{equation*}
Comparing how the similarity measures behave in the data and under the fitted model shows how much of the observed association activity alone can produce.

\paragraph{How much association activity explains.} Since model and data share the popularities, every similarity measure of Section~\ref{sec:setting} is, for a given pair, a function of the co-occurrence alone. Comparing the observed co-occurrence $p_{ij}^{\text{emp}}$ with the co-occurrence $p_{ij}^{\text{model}}$ implied by a model therefore compares lift, cosine and Pearson correlation at once. We measure the distance between the two with the Poisson deviance of the co-occurrence counts over all pairs of items~\cite{nelder1972}. Given the co-occurrence of a model $p_{ij}^{\text{model}}$ and the empirical co-occurrence $p_{ij}^{\text{emp}}$, the deviance is defined as
\begin{equation}\label{eq:deviance}
D(p^{\text{model}}, p^{\text{emp}})=2U\sum_{i<j}\Big[p_{ij}^{\text{emp}}\log\frac{p_{ij}^{\text{emp}}}{p_{ij}^{\text{model}}}-\big(p_{ij}^{\text{emp}}-p_{ij}^{\text{model}}\big)\Big],
\end{equation}
which compares each difference with the natural fluctuation of a count of that size.
We use it only as a distance between counts, not as a test. Even a perfect model leaves some deviance, because counts fluctuate. The share of the systematic mismatch of independence that activity removes is
\begin{equation}\label{eq:share}
\frac
{D(p^{\text{ind}}, p^{\text{emp}})-D(p^{\text{act}}, p^{\text{emp}})}
{D(p^{\text{ind}}, p^{\text{emp}})-\EE[D(p^{\text{act}}, p^{\text{sim}})]}
\end{equation}
where $p_{ij}^{\text{ind}}=p_ip_j$ is the co-occurrence under independence, $p_{ij}^{\text{act}}$ is the co-occurrence under the fitted activity model, and $p_{ij}^{\text{sim}}$ is the co-occurrence of a dataset simulated from the fitted activity model. The share is $1$ if the data look like activity plus counting noise, and $0$ if activity explains nothing beyond independence.
We estimate this \emph{noise floor} $\EE[D(p^{\text{act}}, p^{\text{sim}})]$ as the mean deviance of the same model on ten datasets of the same size simulated from it~\cite{efron1979}; the expectation is over such datasets.

\paragraph{Each measure separately.} The deviance compares co-occurrences, and therefore all measures at once. To compare each similarity measure separately, we pool items into groups. For two groups of items $A$ and $B$, a \emph{block pair}, we use ratios of sums:
\begin{equation*}
\begin{gathered}
\lift_{AB}=\frac{\sum p_{ij}}{\sum p_ip_j},\qquad
\cos_{AB}=\frac{\sum p_{ij}}{\sum \sqrt{p_ip_j}},\\
\rho_{AB}=\frac{\sum(p_{ij}-p_ip_j)}{\sum\sqrt{p_i(1-p_i)\,p_j(1-p_j)}},
\end{gathered}
\end{equation*}
where every sum runs over the pairs $i\in A$, $j\in B$ with $i\neq j$. For groups of a single item, they reduce to the measures of Section~\ref{sec:setting}. Single pairs are too noisy for a comparison measure by measure: most of them are never consumed together, and their lift and cosine have no logarithm. Pooling averages this noise out. For a block pair $\{A,B\}$, let $y_{AB}^{\text{emp}}$, $y_{AB}^{\text{act}}$, $y_{AB}^{\text{ind}}$ and $y_{AB}^{\text{sim}}$ be the values of the measure in the data, under the two models and in a dataset simulated from the fitted activity model (on a log scale for lift and cosine). How much $y_{AB}$ fluctuates by chance differs between block pairs, so we measure each difference against $\sigma_{AB}^2=\EE\big[(y_{AB}^{\text{sim}}-y_{AB}^{\text{act}})^2\big]$: the standardized error
\begin{equation*}
Z(y^{\text{model}},y^{\text{emp}})=\frac{1}{|\mathcal P|}\sum_{\{A,B\}\in\mathcal P}\frac{\big(y_{AB}^{\text{emp}}-y_{AB}^{\text{model}}\big)^2}{\sigma_{AB}^2},
\end{equation*}
over the set $\mathcal P$ of block pairs with at least $5$ expected co-occurrences, is close to $1$ when the data look like activity plus noise. As for the deviance, the share of the error of independence removed by activity is
\begin{equation*}
\frac{Z(y^{\text{ind}},y^{\text{emp}})-Z(y^{\text{act}},y^{\text{emp}})}{Z(y^{\text{ind}},y^{\text{emp}})-\EE\big[Z(y^{\text{act}},y^{\text{sim}})\big]}.
\end{equation*}
Both expectations are estimated on the ten simulated datasets of the noise floor; for $\EE[Z(y^{\text{act}},y^{\text{sim}})]$, each dataset is compared using $\sigma_{AB}^2$ from the other nine.

\paragraph{Popularity surface.} The share does not show \emph{which} part of the popularity structure of a measure activity reproduces. Since $h_i$ is fixed by $p_i$, the activity model's value of every measure for a pair is a function of the two popularities alone: the model has no other information about the items. We therefore describe data and model with the same function of popularity. With $x=\log p_A-m$ and $y=\log p_B-m$, where $p_A$ is the mean popularity of block $A$ and $m$ the median of $\log p_A$, we fit by least squares, over the same block pairs,
\begin{equation}\label{eq:surface}
\log y_{AB}=a+b\,(x+y)+c\,xy+d\,(x+y)^2,
\end{equation}
the general symmetric polynomial of degree two. The level $a$ is the value of a pair of median popularity; the slope $b$ is the elasticity in popularity (if $p_Ap_B$ grows by $1\%$, the measure changes by about $b\%$); since $xy=\frac14[(x+y)^2-(x-y)^2]$, a positive $c$ means that a popular--rare pair has a lower value than a balanced pair with the same $p_Ap_B$; $d$ is the curvature. Fitted to the observed values, Equation~\eqref{eq:surface} describes what popularity alone can express; the difference between observed and model coefficients is the popularity structure activity does not produce. Under independence, $a=b=c=d=0$ for lift, and $b=1/2$, $c=d=0$ for cosine.

\section{Experiments}\label{sec:experiments}
\subsection{Datasets}
Table~\ref{tab:datasets} lists the four datasets we considered. In all of them an interaction is binary, and repeated interactions count once. For MovieLens-1M and MovieLens-20M~\cite{harper2015} we count a rating as a consumption only if it is at least $4$, the usual way of turning ratings into implicit feedback~\cite{liang2018,steck2019}. Last.fm (HetRec 2011)~\cite{cantador2011} lists, for each user, their top 50 artists. Amazon Books (Amazon Reviews 2023, 5-core)~\cite{hou2026} keeps users and books with at least 5 reviews. These construction choices shape the distribution of history lengths, and therefore the amount of phantom correlation a dataset can contain.

\begin{table}[t]
\centering
\footnotesize
\setlength{\tabcolsep}{1.8pt}
\begin{tabular}{lrrrrrr}
\toprule
Dataset & Users & Items & Interact. & $\Var(S)$ & $\sum\limits_i\Var(X_i)$ & $\sum\limits_{i\neq j}\Cov(X_i,X_j)$ \\
\midrule
MovieLens-1M & 6{,}038 & 3{,}533 & 575{,}281 & 11{,}024 & 85 & 10{,}939 \\
MovieLens-20M & 138{,}287 & 20{,}720 & 9{,}995{,}410 & 10{,}453 & 67 & 10{,}386 \\
Last.fm & 1{,}892 & 17{,}632 & 92{,}834 & 34 & 47 & $-$12.7 \\
Amazon Books & 776{,}370 & 495{,}063 & 9{,}488{,}297 & 568 & 12 & 555 \\
\bottomrule
\end{tabular}

\caption{Datasets used and their relation with Equation~\eqref{eq:identity}.}
\label{tab:datasets}
\end{table}

On every dataset except Last.fm, users differ strongly in how many items they consume, and by Equation~\eqref{eq:identity} this fixes a large, positive total covariance between items, which the activity model distributes among the pairs. On Last.fm the total covariance is negative, which no model satisfying Assumptions~\ref{as:independence} and~\ref{as:monotonicity} can produce (Corollary~\ref{cor:possible}). Last.fm thus shows what happens when the assumptions fail.
The models are always fitted on all items. The groups of Section~\ref{sec:fitting} are as follows. On MovieLens-1M, the films consumed by at least $\fmMinUsers$ users are compared one by one {(Section~\ref{sec:share})}, and in blocks of $\fmStatBlock$ films consecutive in popularity for the block statistics; the remaining films are pooled into four buckets, equally spaced in $\log p_i$. On the larger catalogs, we sort all items by popularity and cut them into consecutive \emph{popularity groups} that each hold about the same share of all interactions, so that every group carries a comparable amount of data; items with the same popularity always share a group, so a dataset has at most $\catalogGroups$ groups.
The deviance in Equation~\eqref{eq:deviance} is computed on all pairs of items, without pooling; the comparison of each measure separately uses these blocks and groups (Table~\ref{tab:results}).

\paragraph{Use of AI tools.} The analysis code (model fitting, simulations, statistics, tables and figures) was written and run with the assistance of an AI coding assistant (Claude, Anthropic) under the authors' direction. The authors checked it and take responsibility for the correctness of the results.

\begin{table}[t]
\centering
\footnotesize
\setlength{\tabcolsep}{3pt}
\begin{tabular}{lrrrrr}
\toprule
& Pairs & \multicolumn{4}{c}{Blocks} \\
\cmidrule(lr){2-2}\cmidrule(lr){3-6}
Dataset & all & lift & cosine & Pearson & EASE \\
\midrule
MovieLens-1M & 0.81 & $0.99$ & $0.99$ & $0.99$ & $0.79$ \\
MovieLens-20M & 0.73 & $0.98$ & $0.98$ & $0.96$ & $0.75$ \\
Last.fm & 0.00 & $-0.02$ & $0.03$ & $0.00$ & $0.00$ \\
Amazon Books & 0.50 & $0.84$ & $0.84$ & $0.73$ & -- \\
\bottomrule
\end{tabular}

\caption{Share of the systematic error of the independence model removed by the activity model. Deviance over all pairs of items (Equation~\eqref{eq:share}), and standardized error of each measure over the blocks of items used in the figures.}
\label{tab:results}
\end{table}

\begin{table}[t]
\centering
\footnotesize
\setlength{\tabcolsep}{2.5pt}
\begin{tabular}{lrrrrrrrr}
\toprule
& \multicolumn{2}{c}{$a$ (level)} & \multicolumn{2}{c}{$b$ (popularity)} & \multicolumn{2}{c}{$c$ (imbalance)} & \multicolumn{2}{c}{$d$ (curvature)} \\
\cmidrule(lr){2-3}\cmidrule(lr){4-5}\cmidrule(lr){6-7}\cmidrule(lr){8-9}
& obs. & act. & obs. & act. & obs. & act. & obs. & act. \\
\midrule
\multicolumn{9}{l}{\emph{MovieLens-1M}} \\
Lift & $1.09$ & $1.05$ & $-0.18$ & $-0.12$ & $0.11$ & $0.06$ & $-0.02$ & $-0.02$ \\
Cosine & $-2.45$ & $-2.49$ & $0.32$ & $0.38$ & $0.11$ & $0.06$ & $-0.02$ & $-0.02$ \\
Pearson & $-2.83$ & $-2.90$ & $0.25$ & $0.33$ & $0.16$ & $0.08$ & $-0.03$ & $-0.03$ \\
\midrule
\multicolumn{9}{l}{\emph{MovieLens-20M}} \\
Lift & $1.01$ & $1.09$ & $-0.17$ & $-0.16$ & $0.09$ & $0.07$ & $-0.02$ & $-0.02$ \\
Cosine & $-2.22$ & $-2.13$ & $0.33$ & $0.34$ & $0.08$ & $0.05$ & $-0.01$ & $-0.02$ \\
Pearson & $-2.64$ & $-2.50$ & $0.25$ & $0.28$ & $0.11$ & $0.08$ & $-0.02$ & $-0.03$ \\
\midrule
\multicolumn{9}{l}{\emph{Last.fm}} \\
Lift & $0.07$ & $0.01$ & $0.03$ & $0.00$ & $0.11$ & $0.00$ & $-0.01$ & $0.00$ \\
Cosine & $-3.64$ & $-3.70$ & $0.54$ & $0.50$ & $0.11$ & $0.00$ & $-0.01$ & $0.00$ \\
Pearson & -- & $-8.00$ & -- & $0.52$ & -- & $-0.01$ & -- & $0.01$ \\
\midrule
\multicolumn{9}{l}{\emph{Amazon Books}} \\
Lift & $1.73$ & $1.58$ & $-0.11$ & $0.00$ & $0.24$ & $0.01$ & $-0.05$ & $0.00$ \\
Cosine & $-8.22$ & $-8.37$ & $0.39$ & $0.50$ & $0.24$ & $0.01$ & $-0.05$ & $0.00$ \\
Pearson & $-8.32$ & $-8.60$ & $0.37$ & $0.50$ & $0.46$ & $0.01$ & $-0.11$ & $0.00$ \\
\bottomrule
\end{tabular}

\caption{Popularity surface (Equation~\eqref{eq:surface}) of each measure over block pairs, fitted to the observed values (obs.) and to the activity model (act.). Block pairs with negative Pearson correlation are left out (Amazon Books: \amzNonposPearson{} of \amzBlockPairs); on Last.fm they are \lfmNonposPearson{} of \lfmBlockPairs{}, and the observed fit is not shown.}
\label{tab:trend}
\end{table}

\subsection{How Much Association Activity Explains}\label{sec:share}\label{sec:catalog}
Figure~\ref{fig:ml1m} shows the four measures on MovieLens-1M for every pair of the $\fmItems$ films with at least $\fmMinUsers$ positive ratings ($\fmPairs$ pairs, $\fmDenseShare$ of all interactions; only $\fmNeverPct$ of them were never rated together), with the remaining $\fmTailItems$ films pooled into the four buckets; the figures for the other datasets are in the code repository\footnote{\url{https://github.com/Mfavier91/Phantom-Correlation}}. Table~\ref{tab:results} summarizes all datasets. Where users differ strongly in activity, activity explains most of the association: on MovieLens-1M and MovieLens-20M it removes most of the systematic mismatch of independence over all pairs, and almost all of it over blocks, for every measure. On MovieLens-1M, the observed co-occurrence of $\fmActWithinTwo$ of the pairs is within a factor of two of what activity implies, against $\fmIndWithinTwo$ under independence. On Amazon Books, activity removes about half of the mismatch over all pairs, most of which are too rare to tell the models apart, and most of it over blocks, least for Pearson correlation. On Last.fm, where every list is capped at 50 artists and history lengths vary less than under independence (Table~\ref{tab:datasets}), activity removes nothing: the fitted model is essentially the independence model, as Corollary~\ref{cor:possible} anticipates. EASE, trained on the co-occurrences each model expects, follows the same ranking: activity removes most of the error of independence on both MovieLens datasets, although less than for the other measures, since single EASE weights are dominated by sampling noise, and none on Last.fm; on Amazon Books, the item--item matrix is too large to invert.

\subsection{How Association Depends on Popularity}\label{sec:where}
The activity model reproduces the overall shape of the matrices of Figure~\ref{fig:ml1m}: lift is largest among rare films and cosine among popular ones, in the data and in the model, as Proposition~\ref{prop:limits} and Equation~\eqref{eq:fromlift} predict, and Pearson correlation is largest among popular films, while independence has no correlation at all. The popularity surfaces of Table~\ref{tab:trend} make this precise for every dataset. The observed values are well described by a function of popularity alone ($R^2$ between $\trendRtwoMin$ and $\trendRtwoMax$), and activity gets their level $a$ right everywhere. On the two MovieLens datasets, activity also reproduces the shape: lift falls with popularity and is lower for unbalanced pairs, in the data and in the model. The observed lift falls somewhat faster, and depends more on the imbalance, than activity predicts, most visibly on MovieLens-1M, because rare films are associated more than activity allows; cosine and Pearson correlation inherit this. On Amazon Books, the activity model's lift is almost constant: books this rare sit in the first limit of Proposition~\ref{prop:limits}, where activity adds the same association to every pair, so its cosine grows with popularity as under independence ($b\approx1/2$). The observed lift instead falls with popularity as steeply as on MovieLens, and depends on the imbalance more than on any other dataset: books of similar popularity are read together, and heavy readers drift to niche books, more than activity allows (Appendix~\ref{app:datasets}). On Last.fm, activity produces no popularity structure at all, while the observed lift grows with popularity. In every dataset, what activity does not explain is itself structured by popularity.

\subsection{What Activity Does Not Explain}\label{sec:notexplained}
Activity sets the level of association of a typical pair and, on MovieLens, most of its dependence on popularity; it does not say which pairs stand out. On MovieLens-1M, the Spearman correlation between the observed and the model lift over all pairs is only $\fmActSpearman$, and only the data contain negatively correlated pairs: the ranking of pairs within the activity baseline is where preferences live.

\section{Discussion}\label{sec:discussion}

The picture that emerges is that of a \emph{baseline}. User activity sets, for every pair of items, a level of association that depends on the popularity of the two items, and genuine preferences appear as deviations from it. On MovieLens-1M, the baseline is most of what we observe for a typical pair; the ranking of pairs within that baseline is where taste shows. Where the activity model falls short, what remains is visibly structured by popularity and genuine preferences.

Before continuing, we want to be clear about what we are \emph{not} claiming. We are not claiming that collaborative filtering signals are fake: activity explains the level of association, not which pairs stand out. We are not proposing a statistical test: the activity model is not a hypothesis to reject, but a mechanism whose consequences we want to understand. And we are not claiming that activity is the only source of popularity bias.

What we \emph{are} claiming is that every similarity measure computed on implicit feedback mixes two things, and that one of them, the activity baseline, has a precise mathematical form.
This matters for the choice of a similarity measure. Cosine similarity, the default in item-based collaborative filtering~\cite{sarwar2001,linden2003}, inherits a baseline that grows with popularity: whenever recommendation scores aggregate similarities, popular items receive a head start that has nothing to do with taste~\cite{jannach2015,klimashevskaia2024}. Lift and PMI inherit a baseline that is largest for rare items, which is consistent with their known tendency to favor rare items~\cite{levy2015}. Whether and how much this baseline biases the final recommendations is the natural next question, and the subject of ongoing work.

Where activity alone is not enough, as on Amazon Books and Last.fm, users differ not only in how much they consume, but also in how mainstream they are~\cite{bauer2019}; modeling that second dimension is beyond the scope of this paper.

\section{Conclusion}\label{sec:conclusion}
We started with a streaming service where nobody likes anything together, and every pair of films looked similar anyway. The reason was simple: some users watch more than others. We showed that this simple reason has a precise form. Users who differ only in activity must differ along one axis, along which the odds of every item grow by the same factor; this makes every pair of items positively associated, fixes the total amount of item-item covariance, and shapes it according to item popularity. On real data, this activity baseline accounts for the typical association between films, while genuine taste lives in the deviations from it. Every time we read item-item similarity as a sign of shared taste, we should first ask how much of it is just activity.

\appendix
\section{Proofs and Further Results}\label{app:proofs}
\paragraph{Proof of Proposition~\ref{prop:positive}.} Write $f(\theta)=\PP(X_i=1\mid\theta)$ and $g(\theta)=\PP(X_j=1\mid\theta)$. By the law of total covariance, $\Cov(X_i,X_j)=\EE_\theta[\Cov(X_i,X_j\mid\theta)]+\Cov_\theta(f(\theta),g(\theta))$, and the first term vanishes because the items are independent given $\theta$ (Assumption~\ref{as:independence}); this gives the equality. Let $\theta'$ be an independent copy of $\theta$. Expanding the product shows that
\begin{equation*}
\Cov_\theta(f(\theta),g(\theta))=\tfrac12\,\EE\big[(f(\theta)-f(\theta'))(g(\theta)-g(\theta'))\big].
\end{equation*}
By Assumption~\ref{as:monotonicity}, the functions $f$ and $g$ are non-decreasing, so the two factors never have opposite signs, and the expectation is non-negative.

If $\theta\indep X_i$, the function $f$ is almost surely constant and the expectation is zero; the same holds if $\theta\indep X_j$. Conversely, suppose that neither $f(\theta)$ nor $g(\theta)$ is almost surely constant. Since $f$ is non-decreasing, there is a value $c_f$ such that the sets $L_f=\{f<c_f\}$ and $U_f=\{f\geq c_f\}$ both have positive probability, and every point of $L_f$ lies below every point of $U_f$; define $c_g$, $L_g$ and $U_g$ in the same way. The sets $L=L_f\cap L_g$ and $U=U_f\cap U_g$ still have positive probability, because $L_f$ and $L_g$ are both lower segments of the range of $\theta$, and $U_f$ and $U_g$ both upper segments. Whenever $\theta\in L$ and $\theta'\in U$, both factors are strictly positive, which happens with probability $\PP(\theta\in L)\,\PP(\theta\in U)>0$; hence the expectation is strictly positive. \qed

\paragraph{Proof of Proposition~\ref{prop:limits}.} The popularity $p_i=\EE_\theta[\logistic(h_i+\theta)]$ is continuous and strictly increasing in $h_i$, with limits $0$ and $1$ as $h_i\to-\infty$ and $h_i\to+\infty$, so we can prove the three claims in terms of $h_i$. Write $\gamma=e^\theta$ and $f_i(\theta)=\logistic(h_i+\theta)$. (1) The function $e^{-h}\logistic(h+\theta)=\gamma/(1+e^{h}\gamma)$ increases to $\gamma$ as $h\to-\infty$. By monotone convergence, numerator and denominator of Equation~\eqref{eq:liftmodel}, multiplied by $e^{-h_i-h_j}$, converge to $\EE[\gamma^2]$ and $\EE[\gamma]^2$, and $\EE[\gamma^2]/\EE[\gamma]^2=1+\Var(\gamma)/\EE[\gamma]^2$. (2) $f_i(\theta)\to1$ for every $\theta$; by bounded convergence, numerator and denominator both converge to $\EE[f_j(\theta)]$. (3) Since $\partial f_i/\partial h_i=f_i(1-f_i)$ and these derivatives are bounded by $\tfrac14$, we can differentiate Equation~\eqref{eq:liftmodel} under the expectation:
\begin{equation*}
\frac{\partial\log\lift_{ij}}{\partial h_i}=\frac{\EE[f_i(1-f_i)f_j]}{\EE[f_if_j]}-\frac{\EE[f_i(1-f_i)]}{\EE[f_i]}.
\end{equation*}
Let $Q$ be the distribution of $\theta$ among users who consumed item $i$, that is, $P$ reweighted by $f_i$. Dividing numerator and denominator of both fractions by $\EE[f_i]$, the derivative becomes
\begin{equation*}
\frac{\partial\log\lift_{ij}}{\partial h_i}=-\frac{\Cov_Q(f_i,f_j)}{\EE_Q[f_j]}.
\end{equation*}
Since $f_i$ and $f_j$ are increasing, the covariance is non-negative by the argument in the proof of Proposition~\ref{prop:positive}, and positive unless $\theta$ is almost surely constant, which is exactly the case in which $X_i$ and $X_j$ are independent (Proposition~\ref{prop:positive}). The same holds for $h_j$ by symmetry. In words, the derivative is $\EE[1-f_i(\theta)\mid X_i=1,X_j=1]-\EE[1-f_i(\theta)\mid X_i=1]$: learning that a user also consumed $j$ makes them look more active, and for a more active user item $i$ is closer to saturation. \qed

\paragraph{Drift towards the tail.} Proposition~\ref{prop:competition} says that a user who consumes more items is more likely to consume every item. It does not say that all items gain equally.
\begin{proposition}[Drift towards the tail]\label{prop:drift}
Under the activity model, let item $i$ be less popular than item $j$, so that $h_i<h_j$. Then the ratio
\begin{equation*}
\frac{\PP(X_i=1\mid S=k)}{\PP(X_j=1\mid S=k)}
\end{equation*}
is non-decreasing in $k$, from $e^{h_i-h_j}$ for $k=1$ to $1$ for $k=n$. Consequently, the expected mean popularity of the items in a basket, namely $\EE\big[\tfrac1k\sum_{i\in[n]}p_iX_i\mid S=k\big]$, is non-increasing in $k$.
\end{proposition}
\begin{proof}
Write $a=e^{h_i}$, $b=e^{h_j}$, and let $E_m$ be the $m$-th elementary symmetric polynomial of the values $e^{h_l}$ of the other $n-2$ items $l\neq i,j$ ($E_m=0$ for $m<0$ and $m>n-2$). By Equation~\eqref{eq:conditional}, the baskets of size $k$ that contain $i$ have total weight $a(E_{k-1}+bE_{k-2})$, and those that contain $j$ have total weight $b(E_{k-1}+aE_{k-2})$, with the same normalizing constant. Hence, for $1\leq k\leq n-1$, the ratio is
\begin{equation*}
\frac{a\,(E_{k-1}+bE_{k-2})}{b\,(E_{k-1}+aE_{k-2})}=\frac ab\cdot\frac{1+b\,t_k}{1+a\,t_k},\qquad t_k=\frac{E_{k-2}}{E_{k-1}}.
\end{equation*}
Since $b>a$, this expression increases with $t_k$, and $t_k$ is non-decreasing in $k$ because the elementary symmetric polynomials of positive numbers satisfy Newton's inequalities $E_{k-1}^2\geq E_kE_{k-2}$~\cite{hardy1952}. For $k=1$ we have $t_1=0$, and for $k=n$ both probabilities equal $1$. For the consequence, the shares $s_i(k)=\PP(X_i=1\mid S=k)/k$ sum to $1$ over the items, since a basket of size $k$ contains $k$ items. By the first part, the ratio $s_i(k)/s_j(k)$ is non-decreasing in $k$ whenever $i$ is less popular than $j$: as $k$ grows, the shares move towards less popular items in the likelihood ratio order, which implies first-order stochastic dominance~\cite{shaked2007}. Since the popularity $p_i$ is non-decreasing in $h_i$, its expectation under $s(k)$, which is the expected mean popularity of the basket, is non-increasing in $k$.
\end{proof}
Heavy users consume a less mainstream mix of items, even if nobody has a taste for the tail. The reason is that popular items saturate: a heavy user has already consumed most of them, so additional activity can only go to rarer ones. A drift towards the tail with activity is therefore not, by itself, evidence of a taste for niche items; only a drift faster than the one in Proposition~\ref{prop:drift} is. Amazon Books shows such a faster drift (Appendix~\ref{app:datasets}).

\section{Initialization of the Fit}\label{sec:fitdetails}
Expectation-maximization needs an initial estimate $h_i^{(0)}$ of the item parameters and an initial distribution $\PP(\theta)$ of activity across users. The item parameters are initialized as if items were independent, while the distribution of activity is a discrete distribution on $G=200$ fixed activity levels $\theta_1<\dots<\theta_G$, initialized with equal weights:
\[
  h_i^{(0)}:=\logit(p_i),
  \qquad \omega^{(0)}_g= 1/G.
\]
The extreme activity levels $\theta_1$ and $\theta_G$ are chosen by solving the equations
\begin{equation*}
  \begin{aligned}
   \sum_i\logistic(h_i^{(0)}+\theta_1 + 1.5)&=\min_u S_u,\\
    \sum_i\logistic(h_i^{(0)}+\theta_G - 1.5)&=\max_u S_u.
  \end{aligned}
\end{equation*}
In other words, we find the activity levels that correspond to the shortest and longest observed histories, and widen the range by $1.5$ on the log-odds scale for good measure. The remaining levels are equally spaced between $\theta_1$ and $\theta_G$, and stay fixed during the fit. We check afterwards that the fitted distribution puts no weight on $\theta_1$ and $\theta_G$, which would signal an insufficient range. We stop when the log-likelihood in Equation~\eqref{eq:loglik} changes by less than a relative $10^{-10}$ and every item popularity is reproduced to within a relative $10^{-6}$.

\section{Other Datasets}\label{app:datasets}
The figures for these datasets are in the code repository.

\paragraph{MovieLens-1M: the rarest films.} Among the rarest films ($\fmRareBucketUsers$ positive ratings, $\fmRareBucketItems$ films, $\fmRareBucketShare$ of all interactions), the observed association is $\fmRareWithinAct$ times what activity alone implies: rare films are consumed together far more than activity explains. This excess, however, rests on very few users: only $\fmRareContribUsers$ users rated two of these films, and their contributions are so uneven that they amount to about $\fmRareEffUsers$ users of equal weight.
\paragraph{MovieLens-20M.} On MovieLens-20M, activity removes a share $\mltwfDevShare$ of the systematic mismatch of the independence model, over all pairs of films (the figures in the repository use $\mltwfGroups$ popularity groups). Relative to the activity model, the observed association is close to the model everywhere except the least popular groups, where it is higher than activity allows; the last group pools all films with at most $\mltwfLastGroupMaxUsers$ positive ratings. Heavy users also rate less popular films than the activity model predicts (mean absolute error $\mltwfProfAct$ in $\log_{10}$ popularity; see the Amazon paragraph below): the part that activity does not explain is structured by popularity.

\paragraph{Amazon Books.} Amazon Books shows the clearest popularity structure beyond activity. Heavy readers review much less popular books than light readers, and the activity model cannot reproduce this drift: under it, the popularity of a user's books barely changes with the number of books they review (mean absolute error $\amzProfAct$ in $\log_{10}$ popularity). Its curve coincides with that of the independence model, and this is no accident: by Equation~\eqref{eq:conditional}, given its size, a basket has the same form under both models, with odds that are almost proportional for books this rare, so activity decides how many books a user reviews, not which ones. The drift that Proposition~\ref{prop:drift} allows is real but small. The consequence shows in the association. Within popularity groups, the observed lift is U-shaped, high for rare \emph{and} for popular books, while the activity model gives an almost flat curve. Over the whole catalog, books of similar popularity are consumed together more than activity allows, and books of very different popularity less: the observed lift is highest along the diagonal of the block matrix, while the activity model's is almost uniform. Over all pairs of books, activity removes only a share $\amzDevShare$ of the systematic mismatch of independence. The correlation row tells the same story: the activity model spreads correlation along the rows of the most popular books, while in the data it stays among popular books. Activity gets the overall amount of association right, but not where it sits: readers differ not only in how much they read, but also in how mainstream their choices are, which the activity model, by construction, cannot express.
\paragraph{Last.fm.} Last.fm is the boundary case. Every list is capped at 50 artists, the degrees are under-dispersed (variance ratio $\lfmVarRatio$), and the condition in Equation~\eqref{eq:overdispersion} fails: there is almost no room for phantom correlation, and activity removes none of the systematic mismatch of the independence model ($\lfmDevShare$). The observed association is again structured by popularity: correlation is positive among the most popular artists and slightly negative between popular and less popular ones, a pattern that no single activity dimension can produce, since it only allows non-negative association.

\end{document}